\RequirePackage{fix-cm}
\documentclass[twocolumn]{svjour3}          
\smartqed  
\usepackage{graphicx}
\usepackage{amssymb, amsmath}

\begin{document}

\title{On the non-integrability and dynamics\\ of discrete models of threads
}


\author{Valery Kozlov         \and
        Ivan Polekhin 
}


\institute{V. Kozlov \at
              Steklov Mathematical Institute of Russian Academy of Sciences, Moscow, Russia \\
              \email{kozlov@pran.ru}           
           \and
           I. Polekhin \at
              Steklov Mathematical Institute of Russian Academy of Sciences, Moscow, Russia \\
              \email{ivanpolekhin@mi-ras.ru}
}

\date{Received: date / Accepted: date}

\maketitle

\begin{abstract}
In the paper, we study the dynamics of planar $n$-gons, which can be considered as discrete models of threads. The main result of the paper is that, under some weak assumptions, these systems are not integrable in the sense of Liouville. This holds for both completely free threads and for threads with fixed points that are placed in external force fields. We present sufficient conditions for the positivity of topological entropy in such systems. We briefly consider other dynamical properties of discrete threads and we also consider discrete models of inextensible yet compressible threads.

\keywords{Inextensible thread \and Non-integrability \and Topological entropy \and Discrete model \and Planar linkage}
\end{abstract}

\section{Introduction}
\label{intro}

The study of the motion of a flexible inextensible thread has historically been a popular mechanical problem (see, for instance, the classical book by P. Appell \cite{appell1904traite}). The importance of this topic stems from the role that threads, tethers, ropes, chains and other string-like elements play in a wide range of industrial applications. During most of the twentieth century, these elements were commonly used in light and marine industry  and in instrument engineering \cite{minakov1,schedrov1,hearle1,alekseev1,yakub1,merkin1,irvine1,cost1}. Later, in the second half of the twentieth century, it became understood that tethers can be also useful in space applications \cite{pearson1,bolotina1,bainum1,modi1,bekey1,belets1,netzer1,ander1,modi2,levin1,cosmo1}. Here, it is worth mentioning the monographs by one of the pioneers in this area, V.V. Beletsky with co-authors \cite{belet1,belet2,belet3}. Recently, interest in studying thread dynamics has been revived by the so-called `chain fountain' \cite{biggins1,pf1,kugush1,martins1}.

Usually a thread is considered as an infinitely dimensional system. Therefore, its purely dynamical behavior is difficult to study and only stationary or quasi-stationary configurations can be considered analytically. Moreover, rigorous study of thread dynamics inevitably leads to the consideration of generalized solutions, which is caused by the singularities (cusps) that may appear during the motion of a thread \cite{kozlov2020}.

In this paper, we consider finite dimensional planar models of threads and we present some qualitative results concerning the essentially dynamical properties of the systems. In  most of our results on the non-integrability, we do not assume that the threads are free and move by inertia: we consider threads both in external force fields and threads with fixed points. Even though our models are finite, we do not impose any limitations on the number and the masses of the elements that make up our system.

To be more precise, we will consider the following model of an inextensible and incompressible thread: a finite number ($n$) of rigid segments of equal lengths $l$; the segments form a broken line and are connected by planar hinges; and, masses $m_1,\dots,m_{n+1}$ are located in the endpoints of the segments. It is possible to consider threads with fixed points and threads that form closed contours. In the latter case, the first mass point always coincides with the last one. We will also briefly consider discrete models of compressible threads, which will be specified below.

Our model of an incompressible thread could be called `a Poinsot thread' after Louis Poinsot, who proposed to consider rigid balls strung on a thread to give an interpretation for negative tension that may occur in some equilibrium configurations of a flexible thread \cite{appell1904traite}. 

As mentioned earlier, it is possible to consider various constraints imposed on a thread. For instance, one can fix some points of the system. The following five configurations can be considered as basic from the point of view of possible applications:
\begin{enumerate}
    \item[(\textit{i})] A thread with fixed endpoints (broken line with fixed points),
    \item[(\textit{ii})]  A closed thread with a fixed point ($n$-gon with a fixed point),
    \item[(\textit{iii})] A closed thread (planar $n$-gon),
    \item[(\textit{iv})] A thread with one fixed endpoint ($n$-link pendulum),
    \item[(\textit{v})] A free non-closed thread.
\end{enumerate}

All these systems can be considered as free, in the sense that there are no external forces acting on the system, and as threads in external potential force fields. One can also assume that there are internal forces acting between the masses or the segments of the thread. For instance, these forces can model various elastic properties of our mechanical system.

The topology of the configuration space of the corresponding discrete system will play a key role in our considerations. For cases (\textit{iv}) and (\textit{v}), the structure of the configuration space can be easily understood: it is either an $n$-dimensional torus, or a direct product of an $n$-dimensional torus and a group of parallel translations of the plane (which is isomorphic to $\mathbb{R}^2$). When these systems move by inertia (i.e., the only forces acting on the system are the forces of reaction), we have natural Noetherian first integrals. In case (\textit{iv}), this first integral is the kinetic moment w.r.t. the fixed point. The configuration space obtained after the corresponding reduction is an $(n-1)$-dimensional torus.  In case (\textit{v}), the group of symmetries coincides with the symmetries of the Euclidean plane and we have three Noetherian integrals. After the reduction, we again obtain a system on an $(n-1)$-dimensional torus. For the cases (\textit{i}) --- (\textit{iii}), that will be our main objects of study, the topology of the configuration space can be more complex. 

The rest of this paper is structured as follows. First, we recall some results on the integrability of Hamiltonian systems and we present auxiliary results concerning the topology of planar $n$-gons. In the next section, we show that threads described by models (\textit{i}) --- (\textit{iii}) are not integrable in the class of real analytical functions provided that some natural assumptions hold. We also present some geometrical results concerning the dynamics of these systems. Then, we discuss the question of the positivity of the topological entropy for our models. In the conclusion, we briefly consider several related problems, including possible generalization of our results to higher dimensions and possible models for compressible threads and their properties.

\section{Auxiliary results and definitions}
\subsection{Integrability and non-integrability}
A triple $(S, \omega, H)$, where $S$ is a $2n$-dimensional smooth manifold, $\omega$ is a symplectic structure on $S$ and $H \colon S \to \mathbb{R}$ is a smooth function, is called a Hamiltonian system. A smooth function $F \colon S \to \mathbb{R}$ is called a first integral of the system $(S, \omega, H)$ if
$$
\{F, H\} \equiv 0.
$$
Here $\{ \cdot, \cdot\}$ is the Poisson bracket corresponding to the symplectic structure. We say that system $(S, \omega, H)$ is Liouville integrable (or simply, integrable) if
\begin{enumerate}
    \item There are $n$ first integrals $F_1 = H, \dots, F_n \colon S \to \mathbb{R}$;
    \item These functions are independent, that is, almost everywhere on $S$, $1$-forms $dF_1, \dots dF_n$ are linearly independent;
    \item $\{F_i, F_j\} \equiv 0$ for any $i$ and $j$.
\end{enumerate}
Everywhere below we will consider analytic Hamiltonian systems: manifold $S$ is an analytic manifold, $H$ is a real analytic function. The system is analytically integrable if all of the functions $F_i$ are analytic.

In our considerations, we will assume that $S$ is a cotangent bundle of an $n$-dimensional manifold $M$, that is, $S = T^*M$. By $q$ we will denote local coordinates on $M$ and by $p$ we denote local coordinates on $T_qM$. In particular, in these coordinates the Poisson bracket has the standard form
$$
\{ F,G \} = \sum\limits_{i=1}^n \left( \frac{\partial F}{\partial q_i}\frac{\partial G}{\partial p_i} - \frac{\partial F}{\partial p_i}\frac{\partial G}{\partial q_i} \right).
$$
A more detailed exposition of Hamiltonian mechanics can be found, for instance, in \cite{akn,abraham}. The problem of integrability of Hamiltonian systems is discussed in detail in \cite{kozlovsymm}.

The main tool that will be used here to prove the non-integrability of our system is the following theorem of I.A.\,Taimanov on the non-integrability of geodesic flows \cite{taim1,taim2}.
\begin{theorem}
Given a geodesic flow on an $n$-dimensional closed analytic manifold $M$ with an analytic Hamiltonian function $H \colon M \to \mathbb{R}$. If
\begin{align}
\label{eq1}
\mathrm{dim}\, H_1(M,\mathbb{Q}) > n,
\end{align}
then there are no functions $F_2, \dots, F_n \colon M \to \mathbb{R}$ such that for some energy level $F_1 = H = \mathrm{const} > 0$ we have
\begin{enumerate}
    \item Functions $F_1,\dots,F_n$ are analytic and $\{F_i, F_j\} = 0$ for any $i,j$ in a neighborhood of the level set $H = \mathrm{const}$;
    \item Differentials $dF_1,\dots,dF_n$ are linearly independent on $H = \mathrm{const}$.
\end{enumerate}
\end{theorem}
As a corollary from this theorem, we obtain sufficient conditions for non-integrability of so-called natural Hamiltonian systems. Let us recall that the system is called natural if its Hamiltonian has the form
\begin{align}
\label{eq2}
H(p,q) = H_2(p,q) + H_0(q) = \sum\limits_{i,j=1}^n g^{ij}(q)p_ip_j + H_0(q),
\end{align}
where $H_2$ is a positive definite quadratic form in $p$ (kinetic energy). In accordance to the Maupertuis principle, projections of the solutions of this system onto $M$ can be considered as geodesics of the Jacobi metric. To be more precise, consider level set $H(p,q) = h$, where $h > \max H_0$. The trajectories of the system with the Hamiltonian function $H$ on the level $H = h$ then coincide with the trajectories of the system with the Hamiltonian function $\tilde H$
$$
\tilde H = \sum\limits_{i,j=1}^n \frac{g^{ij}(q)}{h - H_0(q)}p_ip_j
$$
on the level set $\tilde H = 1$. Moreover, if the original Hamiltonian system has a first integral $F$ on a level set $H = h$, then the corresponding geodesic flow of the Jacobi metric has a first integral $\tilde F$ on $T^*M$ (possibly, except for the set $\tilde H = 0$) and 
$$
\tilde F(p,q) = F\left(\frac{p}{\sqrt{\tilde H(p,p)}},q\right).
$$
\begin{corollary}
Given an analytic manifold $M$ such that condition \eqref{eq1} holds and a natural Hamiltonian system on $T^*M$. Then, this system cannot be analytically integrable.
\end{corollary}
 Indeed, if we have $n$ first analytic independent integrals $F_1,\dots,F_n$ for the natural Hamiltonian system, then we obtain functions $\tilde F_1,\dots,\tilde F_n$ satisfying the conditions of Theorem 1.

Theorem 1 can also be applied to more general Hamiltonian systems. Let us have a system with an analytic Hamiltonian function $H$
\begin{align}
\label{eq3}
    H = H_2(p,q) + H_1(p,q) + H_0(q),
\end{align}
where $H_2$ and $H_0$ coincides with the corresponding terms in \eqref{eq2} and
$$
H_1(p,q) = \sum\limits_{i=1}^n b^i(q)p_i.
$$
We will say that system \eqref{eq3} is integrable in the class of polynomial in $p$ first integrals with independent highest degree terms if there exist $n$ first integrals $F_i \colon T^*M \to \mathbb{R}$ ($F_1 = H$) of the form
$$
F_i(p,q) = F^{m_i}_i(p,q) + F^{m_i-1}_i(p,q) + \dots F^0_i(q),
$$
where $F^{m_i}_i, F^{m_i-1}_i, \dots F^0_i$ are homogeneous in $p$ analytic polynomials of degrees $m_i, m_i-1, \dots 0$ correspondingly, $\{ F^{m_i}_i, F^{m_j}_j \} \equiv 0$ for any $i$ and $j$, and $dF^{m_1}_1, \dots dF^{m_n}_n$ are linearly independent almost everywhere.

\begin{corollary}
Given an analytic manifold $M$ such that condition \eqref{eq1} holds and a Hamiltonian system on $T^*M$ with the Hamiltonian function \eqref{eq3}. Then, this system cannot be integrable in the class of polynomials in $p$ with independent highest degree terms.
\end{corollary}

The proof directly follows from the fact that the highest degree terms $F^{m_1}_1 = H_2, \dots F^{m_n}_n$ are first integrals for the geodesic flow with Hamiltonian $\tilde H = H_2$.

It is worth mentioning here that all of the known integrable mechanical systems are integrable in the class of polynomial first integrals with independent highest degree terms.

\subsection{Topology of linkages}

In this section, we present the results on the topological properties of planar linkages (see, for instance, \cite{farber1,farber2}).  

Let us have $n$ planar segments with lengths $l_1$, $l_2$, ..., $l_n$, which form a closed polygon. In the following we assume that for the lengths the following condition holds

\begin{align}
\label{eq4}
\sum\limits_{i=1}^n l_i \nu_i \ne 0, \mbox{ for any } \nu_i = \pm 1.
\end{align}

We will denote the configuration space of the polygon, viewed up to isometries of the Euclidean plane, by $\tilde M$:
\begin{align}
\label{eqMl}
\tilde M = \{ (u_1, ..., u_n) \in \mathbb{S}^1 \times ... \times \mathbb{S}^1 \colon \sum\limits_{i = 1}^n  l_i u_i = 0\} / SO(2).  
\end{align}

Equivalently, we can consider our $n$-gon with one of its sides fixed: the configuration space of this system naturally coincides with $\tilde M$.

\begin{theorem}
\label{th21}
$\tilde M$ is an analytic closed orientable manifold of dimension $n - 3$.
\end{theorem}

\begin{remark}
When condition \eqref{eq4} does not hold, there are a finite number of singularities on $\tilde M$ that correspond to the collinear configurations of the linkage. Note that this condition holds for a generic set of lengths.
\end{remark}

\begin{definition}
Given a polygon with lengths $l_1, l_2, ..., l_n$, we call a subset of its sides $J = \{ i_1, i_2, ..., i_k \}$ short when
$$
\sum\limits_{i \in J}l_i < \sum\limits_{i \not\in J}l_i.
$$
\end{definition}

In the following, we will use the following result on the topology of the configuration space of a linkage.\\

\begin{theorem}
\label{th26}
Given a planar polygon satisfying \eqref{eq1}, let $l_i$ be a side of the maximal length (i.e., $l_i \geqslant l_j$ for any $j$), for every $k \in \{0, 1, ..., n-3\}$, the homology group $H_k(M_l;\mathbb{Z})$ is a free Abelian group of rank $a_k + a_{n-3-k}$, where $a_k$ denotes the number of short subsets of $k+1$ elements containing $l_i$.
\end{theorem}

\begin{corollary}
Let $n = 2r + 1$ and for all $i$ we have $l_i = 1$. Then
\begin{align}
\label{betti}
   b_k(\tilde M)=
\begin{cases}
& C_{n-1}^k, \mbox{ for } k < r - 1,\\
& 2C_{n-1}^{r-1}, \mbox{ for } k = r - 1,\\
& C_{n-1}^{k+2}, \mbox{ for } k > r - 1.
\end{cases} 
\end{align}

\end{corollary}

The following result on the fundamental groups of planar polygons will be used when we will discuss topological entropy. The proof can be found in \cite{schutz1} (see also \cite{schutz2}, where the same technique has been used to calculate the fundamental groups for more complex types of planar linkages).

\begin{theorem}
Let us have a planar polygon and either $l_i = 1$ for all $i$, or $l_i = 1$ for $1 \leqslant i \leqslant n-1$ and $l_n = l$, where $1 \leqslant l < n-1$, $n \geqslant 7$ and condition \eqref{eq4} holds. Then 
\begin{align}
    \label{fund}
    \pi_1(\tilde M) \cong \left\langle a_1,\dots,a_{n-1} \left|\, 
    \begin{aligned}
    &a_k, \mbox{ if } \{k,n\} \mbox{ is not short}\\
    &[a_i,a_j], \mbox{ if } \{i,j,n\} \mbox{ is short}
    \end{aligned}\,
    \right\rangle\right.
\end{align}
where $[a_i,a_j] = a_i^{-1}a_j^{-1}a_ia_j$.
\end{theorem}
In other words, we have a free group with generators $a_1,\dots,a_{n-1}$ and we put $a_k = 1$ if $\{k,n\}$ is not short and we put $[a_i,a_j] = 1$ if $\{i,j,n\}$ is short.

\section{Non-integrability}

In this section we study the non-integrability of models (\textit{i})--(\textit{iii}). For each of these models, we consider two classes of systems, which are different from the dynamical point of view:
\begin{enumerate}
    \item A natural Hamiltonian system. In this case the Hamiltonian of the corresponding thread has the form
    \begin{align}
    \label{eqnat}
    H = H_2(p,q) + H_0(q),        
    \end{align}
    where $H_2(p,q)$ is a quadratic positive definite form in $p$, i.e., $H_2$ is the kinetic energy; $p = (p_1, \dots p_k)$ are the generalized momenta and $q = (q_1, \dots, q_k)$ are the local coordinates on the configuration space of dimension $k$ and $H_0(q)$ is a function on the configuration space which corresponds to the potential forces acting on the system. These forces include both external and internal forces, i.e., for instance, it can be an external force field of gravity or restoring forces caused by the springs located in the joints of the thread. A free thread moving by inertia can be considered as a natural system such that $H_0 \equiv 0$. 
    \item A Hamiltonian system with gyroscopic forces. The Hamiltonian has the form
    \begin{align}
    \label{eqgyro}
    H = H_2(p,q) + H_1(p,q) + H_0(q),
    \end{align}
    where $H_1(p,q) = \sum_{i=1}^k b_i(q) p_i$ and $H_2$ and $H_0$ are defined as above. The term $H_1$ corresponds to so-called gyroscopic forces. For instance, it can be magnetic forces acting on the thread.
\end{enumerate}

\subsection{A thread with fixed endpoints}

Consider the following natural model of a thread with fixed endpoints: a collection of $n$ rigid planar segments of the same length, the segments are pairwise connected by  joints at their endpoints and form a planar broken line between two fixed points. Without loss of the generality we can assume that all segments have unit length. We will denote the distance between the fixed points by $l$ and assume that $l < n$ and $l \notin \mathbb{N}$. Then \eqref{eq4} holds for the closed $(n+1)$-gon such that $l_1 = \dots l_n = 1$ and $l_{n+1} = l$ (Fig. 1).

The above specifies the kinematics of the thread. Its dynamical behavior is determined by the distribution of mass of the thread and by the forces acting on the system.

\begin{figure}[h]
    \centering
    \includegraphics[width=0.49\textwidth]{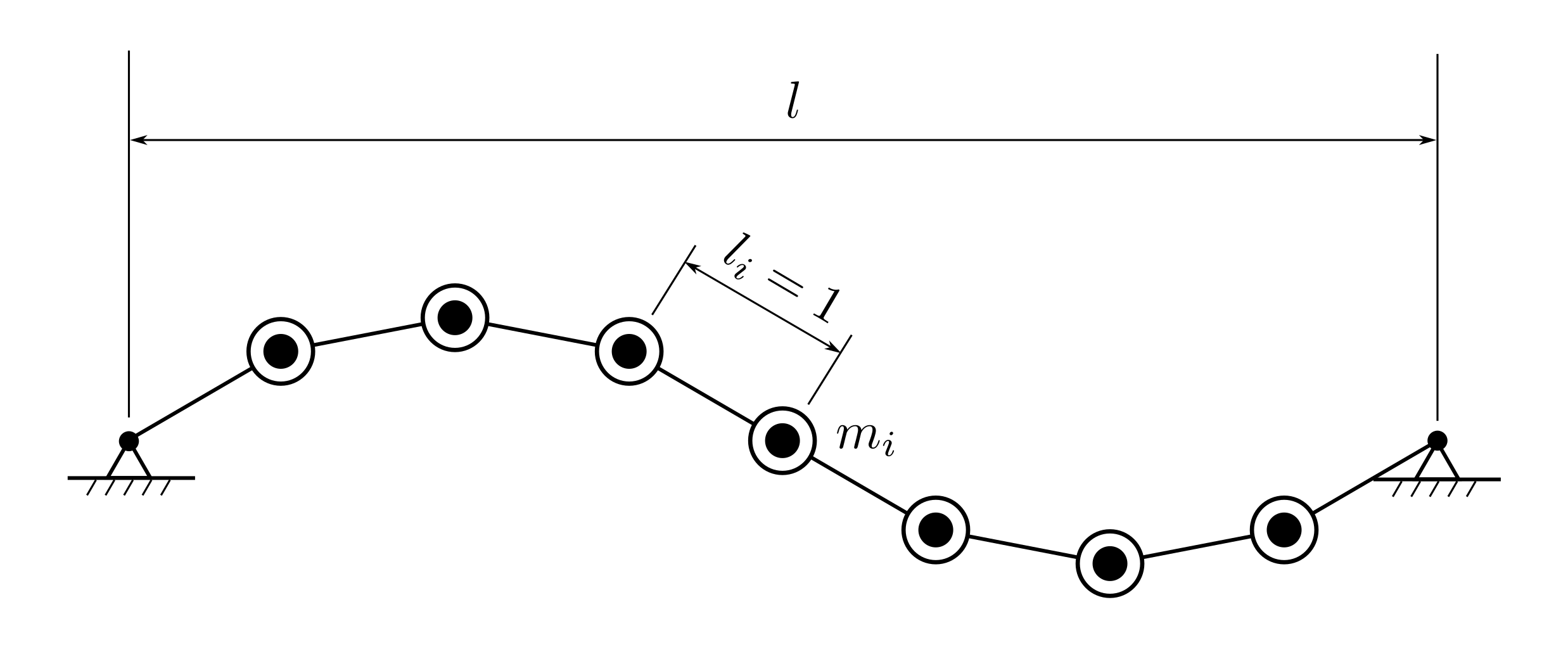}
    \caption{An example configuration of an inextensible and incompressible thread with fixed endpoints (model (\textit{i})).}
    \label{fig1}
\end{figure}

We assume that all mass of the thread is concentrated in the joints, i.e., there are $n-1$ masses $m_i > 0$ such that the $i$-th mass is located in the joint between the $i$-th and the $(i+1)$-th segments of the broken line. We do not consider mass points at the fixed points since they do not affect the dynamics of the system.
When the thread is homogeneous it is natural to assume that $m_i = m_j$ for all $i$ and $j$.

Everywhere below we assume that the thread moves without friction and the system is Hamiltonian. Let the system have $k$ degrees of freedom, that is, the dimension of the configuration space equals $k$.

 Consider the closed polygon such that one of its sides has length $l$ and all other sides are of unit length. From Theorem 2 we obtain that $k = n - 2$. Therefore, we assume that $n \geqslant 4$, since all one-dimensional cases are trivially integrable.

Let us now calculate the first Betti number of this space.

\begin{proposition}
If $0 < l < 1$, then
$$
b_1(M) =
\begin{cases}
n+4, &\mbox{ if } n = 4,5\\
n, &\mbox{ if } n > 5. 
\end{cases}
$$
If  $1 < l < n-2$, $l \notin \mathbb{N}$, then for $n > 4$ we have $b_1(M) = n$ and for $n = 4$ we have $b_1(M) = 8$. If $n-2 < l < n$ then $b_1(M) = 0$.
\end{proposition}
\begin{proof}
The proof is a direct calculation based on Theorem 2. If $0 < l < 1$ and $n = 4$, the total number of elements in the thread is $5$, we have $a_1 = 4$ and $b_1 = a_1 + a_1 = 8$. Similarly, if $0 < l < 1$ and $n = 4$ we have $a_1 = 5$ and $a_2 = 4$, $b_1 = a_1 + a_2 = 9$. If $1 < l < n-2$ and $n = 4$, then $a_1 = 4$ and $b_1 = 2a_1$. If $1 < l < n-2$ and $n > 4$, then $a_1 = n$ and $a_{n-3} = 0$, $b_1 = n$. Finally, for $n-2 < l < n$ we obtain $b_1 = 0$.
\end{proof}
 We see that \eqref{eq1} holds for all $n$ and for all $l < n-2$. From Corollary 1 we have
\begin{proposition}
Let $0 < l < n-2$ and $l \notin \mathbb{N}$, and the Hamiltonian function \eqref{eqnat} of the system is an analytic function, then the system is not analytically integrable.
\end{proposition}
\begin{proof}
Indeed, $\mathrm{dim}(M) = n - 2$ and $b_1(M) \geqslant n$.
\end{proof}

Similarly, from Corollary 2, we have
\begin{proposition}
Let $0 < l < n-2$ and $l \notin \mathbb{N}$, and the Hamiltonian function \eqref{eqgyro} of the system is an analytic function, then the system is not integrable in the class of polynomials in $p$ with independent highest degree terms.
\end{proposition}

In other words, the system of a discrete thread between two fixed points cannot be analytically integrable for $l < n - 2$. This holds for a free thread and for a thread in external or internal force fields. In particular, if we have a thread in a gravity field, then this system is not analytically integrable. If we add a magnetic forces to the system, then this system cannot be integrated in the class of polynomials in $p$ with independent highest degree terms.

Note that for large $n$, that is, when the discrete model of a thread is relatively fine, the condition $l < n - 2$ holds for the most part of the distances between the fixed points. This statement should be understood in the following sense. We can assume that the total length of the thread is fixed and equals $1$. In this case, we can rescale the lengths of the segments and assume that each segment has length $1/n$. Therefore, the system cannot be integrable provided that the distance between the fixed points is less than $(n-2)/n$, i.e., the measure of distances for which the system can possibly be integrable tends to zero as $n$ tends to infinity.

For $n = 4$ we have a system with a two-dimensional configuration space and in this case the non-integrability follows directly from the result for natural Hamiltonian systems with two degrees of freedom \cite{kozlov79}. The proof of this result is based on the existence of a large number of unstable periodic solutions. The asymptotic surfaces of these solutions intersect and form a complex net such that the additional first integral has a constant value at all points of this net. Therefore, taking into account the fact that this integral is analytic, we obtain that this function is a constant.

From the above result on the non-integrability of a thread with fixed endpoints we can obtain the non-integrability for more complex systems having this non-integrable thread as a subsystem. To be more precise, let us have a system such that its configuration space is a direct product of $M$, the configuration space of a non-integrable thread, and $K$, a $k$-dimensional analytic manifold. If $b_1(K) \geqslant k - 1$, then analytic Hamiltonian system with configuration space $M \times K$ cannot be integrable. Indeed, from the K{\"u}nneth theorem, we have
$$
b_1(M \times K) = b_1(M) + b_1(K) \geqslant n + k - 1.
$$
Obviously, $\mathrm{dim}(M \times K) = n + k - 2$ and condition \eqref{eq1} holds. As an example we can consider a non-integrable free thread with a $k$-link pendulum attached to one of the moving joints of the thread. The configuration space of the pendulum is a $k$-dimensional torus and $b_1(\mathbb{T}^k) = k$.

In conclusion of the section, we present a geometrical result concerning the dynamics of the thread in the most general case, that is, in the presence of potential and gyroscopic forces.

First, we shortly recall the correspondence between the Hamiltonian and Lagrangian approaches to the dynamics of mechanical systems. Given a Hamiltonian function of the form \eqref{eqgyro}, we can obtain a Lagrangian $L$ by means of the Legandre transformation:
$$
L(q,\dot q) = \dot q \cdot p - H(p,q), \quad \dot q = \frac{\partial H}{\partial p}.
$$
In the new variables $(q, \dot q)$ we have
\begin{align}
\label{lagrfun}
    L(q, \dot q) = L_2(q,\dot q) + L_1(q,\dot q) + L_0(q),
\end{align}
where, again, $L_2(q,\dot q)$ is a quadratic positive definite form in $\dot q$ and $L_1(q, \dot q)$ is linear in $\dot q$. The dynamics on the tangent bundle $TM$ is defined by the corresponding Lagrange equations.

\begin{proposition}
Given a Lagrangian system with Lagrangian \eqref{lagrfun} and an energy level $h > \max\limits_M (-L_0)$, then any two configurations $q_0, q_1 \in M$ of the thread can be connected by a solution with energy $h$ provided that
\begin{align}
\label{ineqfin}
    4(h + L_0)L_2 > L_1^2
\end{align}
for all $(q,\dot q)$ on the corresponding energy level.
\end{proposition}
\begin{proof}
In accordance to the Maupertuis principle, a path $\gamma \colon [t_0, t_1] \to M$ is a trajectory of a solution of the Lagrangian system iff $\gamma(t)$ is a critical point for the functional $F$
$$
F(\gamma) = \int\limits_{t_1}^{t_2} (2\sqrt{(h + L_0(\gamma))L_2(\gamma, \dot\gamma)} + L_1(\gamma, \dot\gamma))\, dt
$$
in the class of all paths of fixed energy $h$ \cite{akn}. If inequality \eqref{ineqfin} holds, then $F$ defines a Finsler length on $M$ \cite{shen1,shen2}. From the Hopf-Rinow theorem for Finsler manifolds \cite{shen1,shen2}, we have that any two points of $M$ can be connected by a Finsler geodesic. This geodesic corresponds to the desirable solution.
\end{proof}

\subsection{Closed threads}
In this section we will consider models (\textit{ii}) and (\textit{iii}). To a large degree they are similar and both these models will be shown to be non-integrable. However, for model (\textit{iii}), we will impose some additional conditions to prove the non-integrability.

Consider a closed $n$-gon assuming that all its sides have the same unit length and $n$ is an odd number $n \geqslant 5$. Then \eqref{eq4} obviously holds. 
We also assume that one of the points of the $n$-gon is fixed and there are $n-1$ masses $m_i$ located in all non-fixed vertices of the $n$-gon (Fig. 2). Note that we allow self-intersections during the motion of the $n$-gon. Similarly to the case of a thread with two fixed points, we assume that all forces acting on the system are potential and the system is Hamiltonian with the Hamiltonian of the form \eqref{eqnat}. 

\begin{proposition}
Let the Hamiltonian function \eqref{eqnat} of the system be an analytic function. Then the system is not analytically integrable.
\end{proposition}
\begin{proof}
The configuration space $M$ of the system is the direct product of a one-dimensional circle and $\tilde M$ and has dimension $n-2$. From Corollary 3 we obtain that $b_1(\tilde M) = n-1$. Therefore, $b_1(M) = b_1(\tilde M) + 1 = n$ and we can apply Theorem 1.
\end{proof}

Similar result holds for the system with gyroscopic forces.

\begin{proposition}
Let the Hamiltonian function \eqref{eqgyro} of the system be an analytic function. Then the system is not integrable in the class of polynomials in $p$ with independent highest degree terms.
\end{proposition}

Again, the system has the Hamiltonian function of the form \eqref{eqnat} when we consider a totally free thread, that is, there are no external or internal forces acting on the system, except for the forces of reaction. Also, we can consider a thread in an external force field or a thread with interactions between its elements.

Let us now consider a closed thread without a fixed point, that is, a planar $n$-gon with the sides of unit length and $n$ masses $m_i$ located in the vertices. 

\begin{figure}[h]
    \centering
    \includegraphics[width=0.45\textwidth]{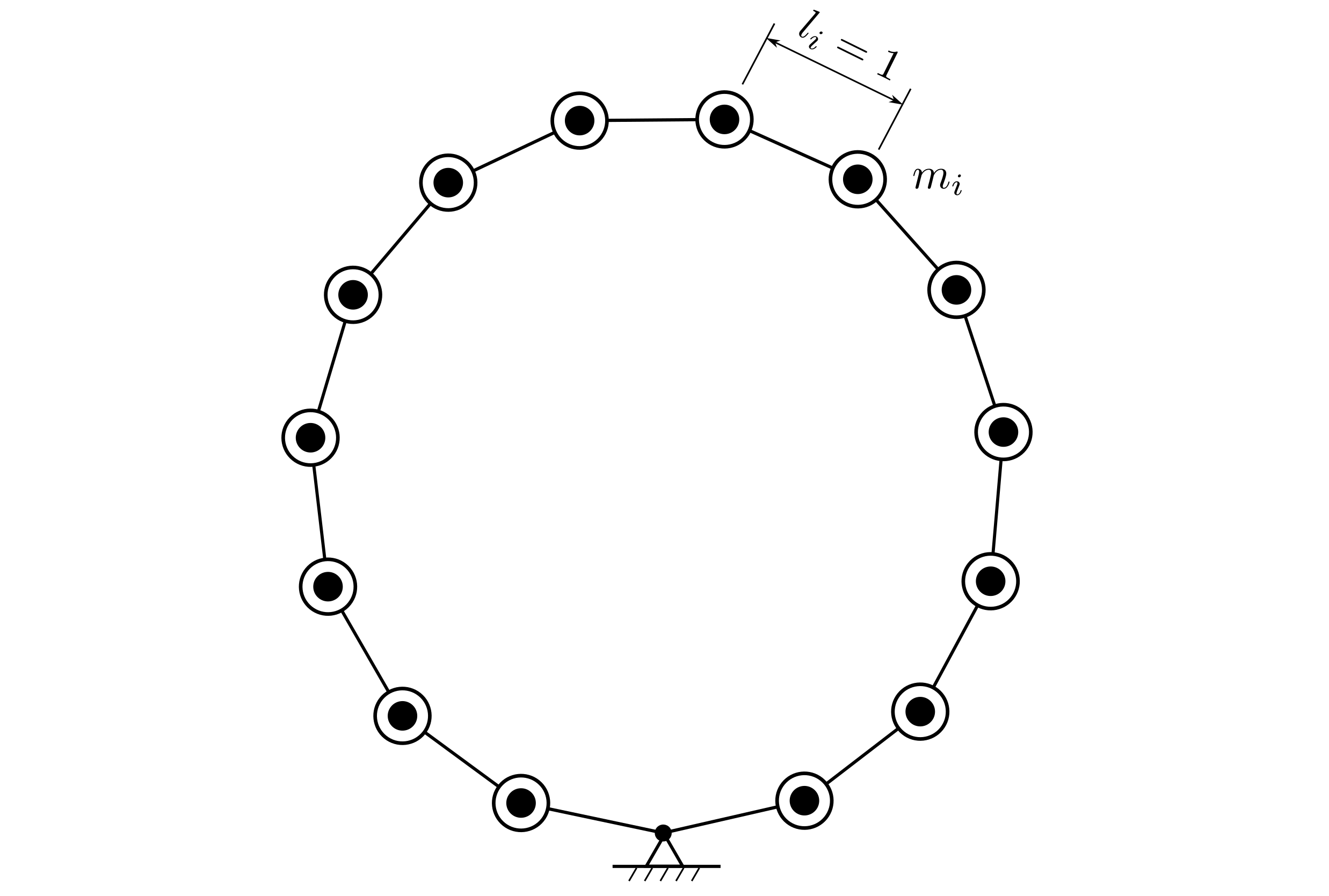}
    \caption{An example configuration of an inextensible and incompressible thread with a fixed point (model (\textit{ii})).}
    \label{fig1}
\end{figure}

The configuration space of this system is not compact and Theorem 1 cannot be applied directly. However, if we assume that there are no external forces acting on the thread, we can consider the reduced system with a compact configuration space. To be more precise, let $x$ and $y$ be the Cartesian coordinates of some mass point of the thread and we consider these coordinates as a part of the set of generalized coordinates. Let the Hamiltonian function of the system has the form \eqref{eqnat}. Since there are no external forces acting on the system, we can conclude that $H$ does not depend on $x$ and $y$. Clearly,
$$
\frac{\partial H}{\partial x} = c_x = \mathrm{const}, \quad \frac{\partial H}{\partial y} = c_y = \mathrm{const}.
$$
After the Routh reduction w.r.t. variables $x$ and $y$ we obtain a Hamiltonian system with the Hamiltonian of the form \eqref{eqgyro}
where $H_1 \equiv 0$ iff $c_x = 0$ and $c_y = 0$. Therefore, similarly to Propositions 2 and 5, we obtain

\begin{proposition}
Let us consider a closed $n$-gon moving on a plane without friction. The lengths of the sides of this $n$-gon equal $1$ ($n$ is an odd number greater than $3$) and masses $m_i$ are located in the vertices of the polygon. Suppose that the only forces acting on the system are the forces of reaction and internal potential forces acting between the elements of the thread. Let $c_x = 0$, $c_y = 0$ and the Hamiltonian function of the reduced system is an analytic function of the form \eqref{eqnat}. Then the reduced system is not analytically integrable.
\end{proposition}

If the initial system contains non-zero terms $H_1$ or at least one of the constants $c_x$ or $c_y$ does not equal zero, then the Hamiltonian of the reduced system takes the form \eqref{eqgyro}. Therefore, we obtain the following result.

\begin{proposition}
Let us consider a closed $n$-gon moving on a plane without friction. The lengths of the sides of this $n$-gon equal $1$ ($n$ is an odd number greater than $3$) and masses $m_i$ are located in the vertices of the polygon. Suppose that the Hamiltonian function $H$ of the system is an analytic function and has the form \eqref{eqgyro} and $H$ does not depend on $x$ and $y$.  Then the reduced system is not integrable in the class of polynomials in $p$ with independent highest degree terms.
\end{proposition}

\subsection{Threads with segments of different length}
Everywhere above we assumed that the segments of the discrete thread are of the same length. Taking into account possible internal forces acting between the segments, we can conclude that this setting allows one to model a broad range of real-life systems. However, for the sake of completeness, we will consider the cases when the segments have different length.

First, we will consider model (\textit{i}). Let $l_i > 0$, $1 \leqslant i \leqslant n$ be the lengths of segments and $l > 0$ be the length between the fixed points. As above, we assume that \eqref{eq4} holds. Inequality $b_1(M) \geqslant n-1$ plays the key role in the proofs of Propositions 2 and 3. From Theorem 3 we have that $b_1(M) \geqslant a_1$. Therefore, if $a_1 \geqslant n-1$, then the system is not integrable. Let $l$ or $l_j$ (for some $1 \leqslant j \leqslant n$) be the side of the maximal length. If there are at least $n-1$ lengths $l_{i_k}$, $1 \leqslant k \leqslant n-1$ (different from the maximal length) such that the pair of lengths $l_{i_k}$ and $l_j$ (or $l$) is a short subset, then the corresponding system is not integrable in the sense of Propositions 2 and 3.

Absolutely similar conditions can be formulated for models (\textit{ii}) and (\textit{iii}). For these cases we have to obtain $b_1(\tilde M) \geqslant n-2$ where $n$ is the number of segments in the thread. Therefore, there should be at least $n-2$ lengths $l_{i_k}$, $1 \leqslant k \leqslant n-2$ (different from the maximal length $l_j$) such that the pair of lengths $l_{i_k}$ and $l_j$ is a short subset.

\section{Topological entropy}

First, let us recall the definition of the topological entropy (see, for instance, \cite{katok1997introduction}). Let $X$ be a compact metric space with a metric $d$ and $f \colon X \to X$ be a continuous map. Consider the following sequence of metrics
$$
d_n(x,y) = \max\limits_{0 \leqslant i \leqslant n-1} d(f^i(x), f^i(y)).
$$
Consider an open ball $B(x,\varepsilon,n) = \{ y \in X \colon d_n(x,y) < \varepsilon \}$. A set $U \subset X$ is an $(n,\varepsilon)$-covering if $X \subset \bigcup_{x \in E} B(x,\varepsilon,n)$. Let $S(\varepsilon,n)$ be the minimal number of elements in an $(n,\varepsilon)$-covering. Put 
$$
h(f,\varepsilon) = \limsup\limits_{n \to \infty} \frac{1}{n} \log S(f,\varepsilon,n).
$$
Then, the topological entropy of the map $f$ is defined as 
$$
h(f) = \lim\limits_{\varepsilon \to 0} h(f,\varepsilon).
$$
The definition of the topological entropy for flows can be expressed in terms of the topological entropy for maps: let us have a flow $\varphi^t \colon \mathbb{R} \times X \to X$, then we put $f = \varphi^1$.

\begin{remark}
This definition is based on a metric structure on $X$. However, it can be shown that this definition does not depend on the choice of the metric, provided that all metrics define the same topology on $X$. A definition that is not based on the metric structure has been given in \cite{adler1965topological}. The definition given above was first given in \cite{dinaburg1971relations}. In addition, the first definition of  entropy for a dynamical system has been formulated by A.N. Kolmogorov \cite{kolmogorov1959entropy}.
\end{remark}

For a geodesic flow on a Riemannian manifold the topological entropy can be defined as follows \cite{ma1997topological}:
$$
h = \lim\limits_{L \to \infty} \frac{1}{L} \log \int\limits_{M \times M} n_L(x,y) dxdy,
$$
where $n_L(x,y)$ is the number of geodesics of lengths no more than $L$ connecting points $x$ and $y$ of manifold $M$.

Positivity of the topological entropy usually corresponds to the complexity of the dynamics of a system. It can also imply the chaotic behavior of a system \cite{downarowicz2014positive}. At the same time, the positivity of topological entropy is not equivalent to the ergodicity and there are non-ergodic systems with a positive topological entropy.

Let us have a geodesic flow on a closed Riemannian manifold. It is known that for some manifolds it is impossible to find a metric with zero topological entropy, that is, for any given smooth metric, the entropy is positive. For instance, if the fundamental group of the manifold is a group of exponential growth, then the topological entropy of the geodesic flow is positive. The details can be found in \cite{dinaburg1971relations,manning1979topological}, where the problem of existence of a metric with zero entropy is considered.

In addition, the following has been proven in \cite{dinaburg1971relations}.
\begin{theorem}
If there exists a metric of negative sectional curvature on a closed manifold, then the geodesic flow on this manifold has a positive topological entropy for any metric.
\end{theorem}

It is known that there exists a metric of negative curvature on any two-dimensional closed manifold of genus greater than one \cite{spivak1970comprehensive}.

In particular, for the previous discrete models of threads, the topological entropy can be proven to be positive when the thread is moving by inertia. To be more precise, given a thread with two fixed endpoints and $n = 4$, the dynamics is described by the geodesic equation provided the motion of the thread is free (i.e., the only forces acting on the thread are the forces of reaction). The metric is given by the kinetic energy of the system and the genus of the configuration manifold is greater than one. Therefore, the topological entropy is strictly positive.

Similar result holds for model (\textit{iii}). However, it is worth mentioning that results about the positivity of geodesic flows can only be applied here for the cases where the constants of the Noetherian integrals equal zero. 

To be more precise, we can conclude that the following results hold for two-dimensional configuration spaces.
\begin{proposition}
Consider a thread with fixed endpoints. Let $n = 4$ and $l_i = 1$ for all $i$. Let the distance between the fixed points be $l < 2$ and condition \eqref{eq4} holds. Suppose that there are massive points with masses $m_i$ located in the joints of the thread and that the only forces acting on the system are the forces of reaction. Then, the topological entropy of this system is positive.
\end{proposition}

\begin{proof}
It is known that the genus $g$ of the surface equals $b_1/2$, that is, for our surface we have $g = 4$.
\end{proof}

\begin{proposition}
Consider a closed thread. Let $n = 5$ and $l_i = 1$ for all $i$. Suppose that there are five massive points with masses $m_i$ located in the joints of the thread and that the only forces acting on the system are the forces of reaction. Also suppose that the constants of three Noetherian first integrals equal zero. Then, the topological entropy of the system (after the Routh reduction) is positive.
\end{proposition}

Some results on the existence of a metric corresponding to zero topological entropy for low dimensional manifolds can be found in \cite{paternain1991entropy,paternain2000differentiable,paternain2006zero}.

In particular, it was proven in \cite{paternain2006zero} that, given a four-dimensional closed manifold $M$ with an infinite fundamental group, it is only possible to find a metric on $M$ with zero topological entropy when the Euler characteristic of $M$ is zero.

As a corollaries from this result, we obtain the following.

\begin{proposition}
Consider a thread with fixed endpoints. Let $n = 6$ and $l_i = 1$ for all $i$. Let the distance between the fixed points is $l < 4$, $l \notin \mathbb{N}$. Suppose that there are massive points with masses $m_i$ located in the joints of the thread and the only forces acting on the system are the forces of reaction. Then, the topological entropy of this system is positive.
\end{proposition}

\begin{proof}
Consider the case when $1 < l < 4$. For the Euler characteristic we have $\chi = b_0 - b_1 + b_2 - b_3 + b_4$ and $b_0 = b_4$, $b_1 = b_3$. Therefore, $b_0 = a_0 + a_4 = 1$, $b_1 = a_1 + a_3 = 6$ and $b_2 = 2a_2 = 30$ (for $1 <l < 2$) or $b_2 = 0$ (for $2 < l < 4$). We see that $\chi \ne 0$. The case $0 < l < 1$ can be considered analogously.
\end{proof}

Now we consider model (\textit{iii}). 

\begin{proposition}
Consider a closed thread. Let $n = 7$ and $l_i = 1$ for all $i$. Suppose that there are massive points with masses $m_i$ located in the joints of the thread and the only forces acting on the system are the forces of reaction. Also suppose that the constants of three Noetherian first integrals equal zero. Then, the topological entropy of the system (after the Routh reduction) is positive.
\end{proposition}
\begin{proof}
From Corollary 3, we have $\chi = 2(1 - 6 + C_6^2) \ne 0$.
\end{proof}

 Note that the fundamental groups of these systems are clearly infinite because their abelianizations, the first homology groups, are infinite.

Note that a result similar to Proposition 12 holds for model (\textit{ii}) if we assume that the only Noetherian first integral equals zero. However, if we do not want to consider the reduced system and, at the same time, we want to obtain a configuration space of dimension $4$, then there should be six segments in the closed contour. If we assume that these segments are of the same length, then the configuration space will not be a smooth manifold.

For an arbitrarily large $n$ (i.e., for the cases when the thread is modeled by a large number of segments), it is also possible to prove that the entropy is positive based on Theorem 4.

\begin{proposition}
Consider a thread with fixed endpoints. Let $n > 5$ and $l_i = 1$ for all $i$. Let the distance between the fixed points is $l$ and $n-4 < l < n-2$, $l \notin \mathbb{N}$. Suppose that there are massive points with masses $m_i$ located in the joints of the thread and that the only forces acting on the system are the forces of reaction. Then, the topological entropy of this system is positive.
\end{proposition}
\begin{proof}
From Theorem 4 we have that $\pi_1(M)$ is free with $n-1$ generators: $\{i,n+1\}$ is always short and, conversely, $\{i,j,n+1\}$ is never short. Hence, $\pi_1(M)$ is a free group, that is, a group of an exponential growth.
\end{proof}

Note, that for the closed thread with equal segments from Theorem 4, we obtain that $\pi_1(\tilde M)$ is commutative. Therefore, we cannot conclude that the entropy is positive. Nevertheless, one can expect the topological entropy to be positive for these systems as well, yet the proof of this fact should follow not from the topological properties of the configuration space, but from the metric properties defined by the distribution of mass of the thread.

\section{Conclusion and Final Remarks}

To the best of our knowledge, the above propositions give the first non-trivial applications of the Taimanov's theorem \cite{taim1,taim2}. Note that, apparently, models (\textit{iv}) and (\textit{v}) are also non-integrable. However, again, this non-integrability does not follow from the topological properties of the configuration space ($(n-1)$-dimensional torus), but follows instead from the metric structure defined by the kinetic energy on this torus.

The next natural question that can be considered is the generalization of these results for non-integrability to the cases of spatial motion of the threads. The homology groups of spacial chains has been obtained in \cite{klyachko1994spatial}. In particular, for a closed $n$-gon where $n = 2k + 1$ and all $l_i = 1$, odd Betti numbers of the configuration space (again, considered up to the symmetries of the Euclidean space) vanish.  Therefore, Theorem 1 cannot be applied and, similarly to models (\textit{iv}) and (\textit{v}), non-integrability does not follow from these topological considerations. Here it is worth mentioning that there is a conjecture \cite{taim3} generalizing Theorem 1 that claims that the system is not integrable if for some $k$
$$
\mathrm{dim}\, H_k(M,\mathbb{Q}) > C_n^k.
$$
If this conjecture is true, that it is also possible to prove the non-integrability of spatial threads.

It is also possible to consider another type of thread, a thread that is inextensible yet can be compressed. In this case one should assume that the distance between two consecutive mass points is not equal to $l_i$, but does not exceed this value. From the mechanical point of view, one can imagine that the mass points are connected not by rigid massless rods, but by inextensible ropes.

In the simplest case when we have only one mass point connected to two fixed points. The motion is assume to be free, that is, there are no external forces acting on the system. Let $L$ be the distance between the fixed points, $l_1$ and $l_2$ be the lengths of the two ropes connecting the mass point to the fixed points (Fig. 3).

This system can be considered as a billiard with a non-smooth boundary. One of the first works where this system was considered for $l_1 = l_2$ is \cite{heller}. Later this case was studied numerically \cite{ree,makino}. It was shown that for almost all distances between the fixed points, the system is not ergodic, since there exist stable periodic trajectories. Apparently, for $l_1 \ne l_2$, the ergodicity of the corresponding billiard systems is not exceptional. To be more precise, in the two-dimensional space of parameters $l_1/l_2$, $L$ there is a set of non-zero measure corresponding to the ergodic systems \cite{chen}. This set is a subset of all systems with the hyperbolic periodic trajectory of period $2$ (this trajectory corresponds to the horizontal periodic motion in Fig. 3). Note, that the stability of the elliptic trajectory of period $2$ has been rigorously established in \cite{kamphorst}, of course, these systems cannot be ergodic.

\begin{figure}[h]
    \centering
    \includegraphics[width=0.4\textwidth]{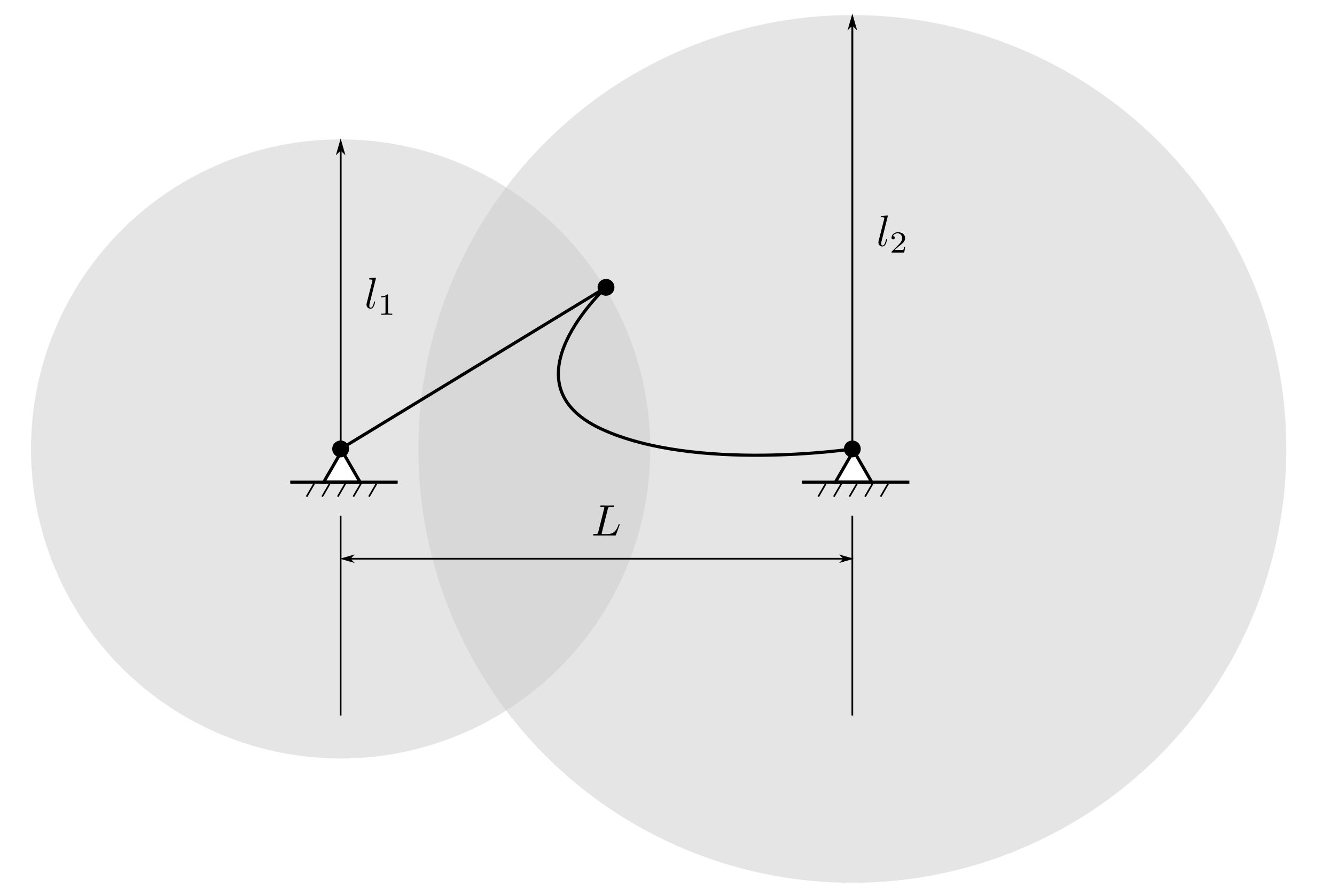}
    \caption{An example configuration of a simple model of an inextensible and compressible thread with two fixed points.}
    \label{fig1}
\end{figure}

The dynamics of two and more mass points connected by inextensible ropes is even more complex and, to the best of our knowledge, has not been studied --- at least numerically --- before.

However, it is possible to obtain a geometric result concerning the dynamics of the compressible threads provided that we consider `almost inextensible' threads. Let us have $n$ massive points moving on a plane without friction and the first and the last points are fixed. We assume that the point with number $2 \leqslant i \leqslant n-1$ interacts with points $i-1$ and $i+1$ and the potential energy of this interaction has the form $U(r_{i-1,i}) + U(r_{i,i+1})$, where $r_{i-1,i}$ and $r_{i,i+1}$ are the distances between the corresponding points and $U(d)$ is a smooth monotonous function such that $U(d) \equiv 0$ for $d \leqslant 1$ and $U(d) \to +\infty$ as $d \to +\infty$. Let the Lagrangian of the system have the form \eqref{lagrfun}, that is, we assume that there can be external potential and gyroscopic forces acting on the system. Since the total energy $L_2 - L_0 = h$, where $h \in \mathbb{R}$, does not change along the solutions of the considered system and $L_2 \geqslant 0$, then for any solution we have $-L_0(q) + h \geqslant 0$. Therefore, for a given energy $h$, the possible motion area $B_h$ is defined as follows
\begin{equation}
\label{eqregion}
B_h = \{ q \colon -L_0(q) + h \geqslant 0 \}.
\end{equation}
If $U(d)$ is a rapidly increasing function, then, for a fixed $h$, the maximum distance between any two consecutive points is close to $1$, that is, the thread is `almost inextensible'. The following result is proved in \cite{kozpol}

\begin{theorem} Let $B_h$ be a compact region and there are no critical points of $L_0$ at boundary $\partial B_h$. If the inequality $4(h+L_0)L_2 > L_1^2$ is true in $B_h \setminus \partial B_h$ for any $\dot q \ne 0$, then any point inside $B_h$ can be connected with the boundary $\partial B_h$ by a solution of energy $h$.
\end{theorem}

This result in some sense complements Proposition 4: we obtain that any configuration of the thread in the possible motion area can be obtained if we start from the boundary $\partial B_h$. In particular, if $L_1 \equiv 0$, the potential energy of the external forces acting on the system is bounded and $h$ is relatively large, then we can conclude that any configuration in $B_h$ can be obtained from another configuration such that at least one pair of massive points are under tension (the corresponding distance is slightly greater than $1$).

In conclusion, returning to the question of non-integrability, we would like to mention an interesting parallel between 
the non-integrability of threads with inner interactions between the elements, which can be considered as various models for elastic properties of the system, and the classical wave equation describing the motion of an elastic string with fixed endpoints. In contrast to our model of the thread, this equation can be integrated explicitly and the general solution is a sum of the standing waves. The key difference between these two systems is that the wave equation describes the motion of an extensible string. Therefore, it may be useful to consider yet another model based on a planar or spatial polygon with extensible sides. The topology of such systems has been already studied in \cite{farber3}.

\section*{Acknowledgment}
This work was performed at the Steklov International Mathematical Center and supported by the Ministry of Science and Higher Education of the Russian Federation (agreement no. 075-15-2019-1614).



\begin{thebibliography}{}
%
%
\bibitem{appell1904traite}
Appell, P.: Trait{\'e} de m{\'e}canique rationnelle, Volume 2, Gauthier-Villars (1904)

\bibitem{minakov1}
Minakov, A.\,P.: Fundamentals of the thread mechanics [In Russian], The Research Work of the Moscow Textile Institute, Volume 9(1), pp. 1--88 (1941)

\bibitem{schedrov1}
Schedrov, V.\,S.: Fundamentals of the Flexible Thread Mechanics [In Russian], Mashgiz, Moscow (1961)

\bibitem{hearle1}
Hearle, J., Grosberg, P., Backer, S.: Structural Mechanics of Fibers, Yarns, and Fabrics (1969)

\bibitem{alekseev1}
Alekseev, N.\,I.: Statics and Steady Motion of a Flexible String [In Russian].
Legkaja Industrija, Moscow (1970)

\bibitem{yakub1}
Yakubovsky, Y.\,V., Zhivov, V.\,S., Korytysskiy, Y.\,I., Migushov, I.\,I.: Principles of the Yarn Mechanics [In Russian], Legkaya Industriya, Moscow (1973)

\bibitem{merkin1}
Merkin, D.\,R.: Introduction to the Mechanics of a Flexible Yarn [In Russian], Nauka, Moscow (1980)

\bibitem{irvine1}
Irvine, H.\,M.: Cable Structures, The M.I.T. Press, Cambridge, MA, (1981)

\bibitem{cost1}
Costello, G.\,A.: Theory of Wire Rope, Springer Science \& Business Media (1997)

\bibitem{pearson1}
Pearson, J.,: The orbital tower: a spacecraft launcher
using the Earth's rotational energy, Acta Astronautica,
Vol. 2, No. 9/10, p. 785--799 (1975)

\bibitem{bolotina1}
Bolotina, N. E., and Vilke, V. G.,: Stability of the
equilibrium positions of a flexible heavy fiber attached
to a satellite in a circular orbit, Cosmic Research , Vol.
16, No. 4, p. 506--510 (1979)

\bibitem{bainum1}
Bainum, P. M., and Kumar, V. K.,: Optimal control of
the shuttle-tethered system, Acta Astronautica, Vol. 7,
No. 12, p. 1333--1348 (1980)

\bibitem{modi1}
Modi, V.J., Chang-Fu, G., Misra, A. K., and Xu, D. M.: On the control of the space shuttle based tether system, Acta Astronautica, Vol. 9. No. 6--7, p. 437--443
(1982)

\bibitem{bekey1}
Bekey, I.: Tethers open new space options, Astronautics
and Aeronautics , Vol. 21, No. 4, p. 33--40 (1983)

\bibitem{belets1}
Beletskii, V. V., and Levin, E. M.: Dynamics of the
orbital cable system, Acta Astronautica, Vol. 12, No. 5,
p. 285--291 (1985)

\bibitem{netzer1}
 Netzer, E. and Kane, T. R.: An alternate approach to
space missions involving a long tether, Journal of the
 Astronautical Sciences, Vol. 40, No. 3, p. 313--327 (1992)

\bibitem{ander1}
Anderson, L. A.: Tethered elevator design for space
station, Journal of Spacecraft and Rockets, Vol. 29, p.
233--238 (1992)

\bibitem{modi2}
Modi, V.J., Bachmann, S. and Misra, A.K.: Dynamics and
control of a space station based tethered elevator
system, Acta Astronautica, Vol. 29, No. 6, p. 429--449 (1993)

\bibitem{levin1}
Levin, E.M.: Nonlinear oscillations of space tethers,
 Acta Astronautica, Vol. 32, No. 5, p. 405--408 (1994)

\bibitem{cosmo1}
Cosmo, M.L., Lorenzini, E.C.: Tethers in Space Handbook, Smithsonian Astrophysical Observatory (1997)

\bibitem{belet1}
 Beletsky, V.\,V., Levin, E.\,M.: Dynamics of Space Tether Systems. Vol. 83. Univelt Incorporated (1993)

\bibitem{belet2}
H. Troger, A.\,P. Alpatov, V.\,V. Beletsky, V.\,I. Dranovskii, V.\,S. Khoroshilov, A.\,V. Pirozhenko, A.\,E. Zakrzhevskii.: Dynamics of Tethered Space Systems. CRC Press (2010)

\bibitem{belet3}
 Beletsky, V.\,V.: Essays on the Motion of Celestial Bodies. Birkhäuser (2012)

\bibitem{biggins1}
Biggins, J.S., Warner, M.: Understanding the chain fountain, Proceedings of the Royal Society A: Mathematical, Physical and Engineering Sciences 470.2163 (2014)

\bibitem{pf1}
Pfeiffer, F., Mayet, J.: Stationary dynamics of a chain fountain. Archive of Applied Mechanics, 87(9), pp. 1411--1426 (2017)

\bibitem{kugush1}
Gyulamirova N.S., Kugushev E.I.: Stationary form of a moving heavy flexible thread, Vestnik Moskovskogo Universiteta, Seriya 1, Matematika, Mekhanika, (1) pp. 39--43 (2018)

\bibitem{martins1}
Martins, R.: The (not so simple!) chain fountain, Experimental Mathematics, 28(4), pp. 398--403 (2019)

\bibitem{kozlov2020}
 Kozlov, V.\,V.: Isoperimetric inequalities for moments of inertia and stability of stationary motions of a flexible thread, Russian Journal of Nonlinear Dynamics 15.4, pp.\,513--523 (2019)

\bibitem{akn}
Arnold, V.I., Kozlov, V.V.,  Neishtadt, A.I.: Mathematical aspects of classical and celestial mechanics (Vol. 3). Springer Science \& Business Media (2007)

\bibitem{abraham}
Abraham, R., Marsden J.E.: Foundations of Mechanics. Vol. 36. Reading, Massachusetts: Benjamin/Cummings Publishing Company (1978).

\bibitem{kozlovsymm}
Kozlov, V. V.: Symmetries, topology and resonances in Hamiltonian mechanics (Vol. 31), Springer Science \& Business Media (2012)

\bibitem{taim1}
Taimanov, I.A.: Topological obstructions to integrability of geodesic flows on non-simply-connected manifolds, Mathematics of the USSR-Izvestiya, 30(2), pp. 403--409 (1987)

\bibitem{taim2}
Taimanov, I.A.: On topological properties of integrable geodesic flows, Mat. Zametki, 44:2, pp. 283–284 (1988)

\bibitem{farber1}
Farber, M., Schütz, D.: Homology of planar polygon spaces, Geometriae Dedicata 125 pp. 75--92 (2007)

\bibitem{farber2}
Farber, M.: Invitation to Topological Robotics, Volume 8, European Mathematical Society (2008)

\bibitem{schutz1}
Sch{\"u}tz, D.: The fundamental group of planar polygon spaces, http://citeseerx.ist.psu.edu/viewdoc/download\\?doi=10.1.1.364.1337\&rep=rep1\&type=pdf 

\bibitem{schutz2}
Schütz, D.: The isomorphism problem for planar polygon spaces, Journal of Topology 3(3) pp. 713-742 (2010)

\bibitem{kozlov79}
Kozlov, V. V.: Topological obstructions to the integrability of natural mechanical systems, Sov. Math. Dokl., 20 pp. 1413–1415 (1979)

\bibitem{shen1}
Shen, Z.: Lectures on Finsler geometry, World Scientific (2001)

\bibitem{shen2}
Bao, D., Chern, S.-S., Shen, Z.: An introduction to Riemann-Finsler geometry, volume 200, Springer Science \& Business Media (2012).

\bibitem{katok1997introduction}
Katok, A., Hasselblatt, B.: Introduction to the Modern Theory of Dynamical Systems, volume 54, Cambridge University Press (1997)

\bibitem{adler1965topological}
Adler, R. L., Konheim, A. G., McAndrew, M. H.: Topological entropy, Transactions of the American Mathematical Society 114 pp. 309–319 (1965)

\bibitem{dinaburg1971relations}
Dinaburg, E. I.: On the relations among various entropy characteristics of dynamical systems, Mathematics of the USSR-Izvestiya 5 pp. 337--378 (1971)

\bibitem{kolmogorov1959entropy}
Kolmogorov, A.N.: Entropy per unit time as a metric invariant of automorphisms, Dokl. Akad. Nauk SSSR, volume 124, pp. 754–755 (1959)

\bibitem{ma1997topological}
Mãne, R:, On the topological entropy of geodesic flows, Journal of Differential Geometry 45 pp. 74–93 (1997)

\bibitem{downarowicz2014positive}
Downarowicz, T.: Positive topological entropy implies chaos dc2, Proceedings of the American Mathematical Society 142 pp. 137–149 (2014)

\bibitem{manning1979topological}
Manning, A.: Topological entropy for geodesic flows,   Annals of Mathematics 110 pp.567–573 (1979)

\bibitem{spivak1970comprehensive}
 Spivak, M. D.: A comprehensive introduction to differential geometry, Wilmington, DE: Publish or perish (1970)

\bibitem{paternain1991entropy}
Paternain, G.: Entropy and completely integrable Hamiltonian systems, Proceedings of the American Mathematical Society 113 pp.871–873 (1991)

\bibitem{paternain2000differentiable}
Paternain, G.P.: Differentiable structures with zero entropy on simply connected 4-manifolds, Boletim da Sociedade Brasileira de Mateḿatica 31 pp. 1–8 (2001)

\bibitem{paternain2006zero}
Paternain, G.P., Petean, J.: Zero entropy and bounded topology, Commentarii Mathematici Helvetici 81 pp.287–304 (2006)

\bibitem{klyachko1994spatial}
 Klyachko, A. A.: Spatial polygons and stable configurations of points in the projective line, in: Algebraic Geometry and Its Applications, Springer, pp. 67–84 (1994)

\bibitem{taim3}
Taimanov, I. A.: The topology of Riemannian manifolds with integrable geodesics flows, Trudy Matematicheskogo Instituta imeni V.A. Steklova 205 pp. 150-163 (1994)

\bibitem{heller}
Heller, E. J., Tomsovic, S.: Postmodern quantum mechanics., Physics Today 46 pp. 38–46 (1993)

\bibitem{ree}
 Ree, S., Reichl, L.E.: Classical and quantum chaos in a circular billiard with a straight cut, Physical Review E 60, 1607 (1999)

\bibitem{makino}
Makino, H., Harayama, T., Aizawa Y.: Quantum-classical correspondences of the Berry-Robnik parameter through bifurcations in lemon billiard systems, Physical Review E 63, 056203 (2001) 

\bibitem{chen}
Chen, J., Mohr, L., Zhang, H.-K., Zhang, P.: Ergodicity of the generalized lemon billiards, Chaos: An Interdisciplinary Journal of Nonlinear Science 23, 043137 (2013)

\bibitem{kamphorst}
Kamphorst, S. O., Pinto-de Carvalho, S.: The first Birkhoff coefficient and the stability of 2-periodic orbits on billiards, Experimental Mathematics 14 pp. 299–306 (2005)

\bibitem{kozpol}
Kozlov, V., Polekhin, I.: On the covering of a Hill’s region by solutions in systems with gyroscopic forces, Nonlinear Analysis: Theory, Methods \& Applications, 148, pp.138-146 (2017)

\bibitem{farber3}
Farber, M. and Fromm, V.: Homology of planar telescopic linkages, Algebraic \& Geometric Topology, 10(2), pp.1063-1087 (2010)

\end{thebibliography}
\end{document}